\documentclass[11pt, peerreview]{IEEEtran}

\usepackage[dvipdfmx]{graphicx,xcolor}
\usepackage[fleqn]{amsmath}
\usepackage{ascmac}
\usepackage{amsthm}
\usepackage{amssymb}
\usepackage{braket} % \Set
\usepackage[mathscr]{euscript}
\usepackage{subcaption}
\theoremstyle{definition}
\newtheorem{Theorem}{Theorem}
\newtheorem*{Theorem*}{Theorem}
\newtheorem{Definition}[Theorem]{Definition}
\newtheorem*{Definition*}{Definition}

\newtheorem{Lemma}[Theorem]{Lemma}
\newtheorem{Cor}[Theorem]{Corollary}

\newcommand\marginblock{}

\newcommand{\ceil}[1]{\left\lceil#1\right\rceil}
\newcommand{\abs}[1]{\left\lvert#1\right\rvert}

\newcommand{\bvec}[1]{\boldsymbol{#1}}

\graphicspath{{./fig/}}

\begin{document}

\title{On Zero Error Capacity of Nearest Neighbor Error Channels with Multilevel Alphabet}
\author{
	\authorblockN{Takafumi Nakano and Tadashi Wadayama} \\[0.3cm]
	\authorblockA{
		Nagoya Institute of Technology\\
		{\small email: wadayama@nitech.ac.jp }
		\footnote{A part of this work was presented at the International Symposium on Information Theory and Its Applications 2016.}\\
	}
}

\maketitle

\begin{abstract}
This paper studies the zero error capacity
 of the Nearest Neighbor Error (NNE) channels with a multilevel alphabet.
In the NNE channels, 
 a transmitted symbol is a $d$-tuple of elements  in $\{0,1,2,\dots, n-1 \}$.
It is assumed that only one element error to a nearest neighbor element in a transmitted symbol can occur.
The NNE channels can be considered as a special type of limited magnitude error channels,
 and it is closely related to error models for flash memories.
In this paper, we derive a lower bound of the zero error capacity of the NNE channels
 based on a result of the perfect Lee codes.
An upper bound of the zero error capacity of the NNE channels is also derived
 from a feasible solution of a linear programming problem defined based on the confusion graphs 
 of the NNE channels.
As a result,  a concise formula of the zero error capacity
 is obtained using the lower and upper bounds.
\end{abstract}

%=================
\section{Introduction}
%=================

In \cite{shannon1956},
 Shannon defined the {\em zero error capacity} of a memoryless channel with
 finite input and output alphabets
by the supremum of achievable rates of codes
 which have no decoding error.
 There are several reasons to study the zero error capacity~\cite{korner1998}.
 Firstly, the zero error capacity indicates the supremum of the achievable rates
 at which no errors can be tolerated and it thus indicates one of fundamental properties of the channel.
 Secondly, the zero error capacity suggests a property of codes with extremely low decoding error probability
 on an environment which does not allow to lengthen the code length enough.
 Shannon proposed the zero error capacity and, in addition,
 showed the equivalence to the maximum independent set problem
 on an undirected graph called a confusion graph~\cite{shannon1956}.
 Exact evaluation of the zero error capacity of a channel is a challenging problem.
 No general solution that gives the exact value of the zero error capacity has been known.
 Therefore, there have been many studies
 for deriving upper and lower bounds of the zero error capacity on general channels.

For example, Shannon showed that the logarithm of the independence number of a given confusion graph
is a lower bound of the zero error capacity of the corresponding channel~\cite{shannon1956}.
He also presented an upper bound
based on the logarithm of the optimal value of a Linear Programming (LP) problem.
Unfortunately, this upper bound is not so tight for many problems in general.
For example, this LP upper bound  is larger than the zero error capacity
of the channel corresponding to the cycle graph with 5-edges, i.e., the pentagon graph. 
This means the LP upper bound is useless to determine the exact value of the zero error capacity of
the pentagon graph.

 Lov\`asz showed a tighter upper bound than the LP upper bound \cite{lovasz1979}.
 The value of this upper bound is called the {\em Lov\`asz number}, which can be computed
 by solving a semi-definite programming problem.
 It is well known that the Lov\`asz number is equal to the value of the trivial lower bound
 in the case of the pentagon graph, i.e., this means that the Lov\`asz number is
 equal to the exact value of the zero error capacity.
 However, the zero error capacity of the cycle graph with 7-edges remains to be open.
 It is known that, if a confusion graph is a perfect graph,
 the upper and lower bounds shown by Shannon coincide with each other.
 Therefore, the zero error capacity of the channel corresponding to a perfect graph 
 can be exactly evaluated~\cite{lovasz1979}.

In the NNE channels, a transmitted symbol is a $d$-tuple of elements  in $\{0,1,2,\dots, n-1 \}$.
It is assumed that only one element error to a nearest neighbor element in a transmitted symbol can occur.

Figure~\ref{fig:nne_channel} presents the channel transition graph of the NNE channel
in which a transmitted symbol  is  a $(d=2)$-tuple of elements in $\{0, 1, 2\}$,  i.e, $n = 3$.
For example, from the adjacency relations of the symbols shown in Fig.~\ref{fig:nne_channe_adj},
the symbol $(0, 0)$ can be erroneously received as $(0, 1)$ or $(1, 0)$, which are the nearest neighbors of $(0, 0)$.

The NNE channels can be considered as a special type of limited magnitude error channels \cite{cassuto2010}, \cite{elarief2013}
and it is closely related to error models for flash memories.
When $d=1$, the confusion graphs of the NNE channels are perfect graphs.
Therefore the zero error capacity can be exactly evaluated in such cases.
However, if $d \ge 2$, the zero error capacity of these channels has been unknown.

\begin{figure}[htb]
    \begin{minipage}[b]{0.48\linewidth}
        \centering
        \includegraphics[width=30mm]{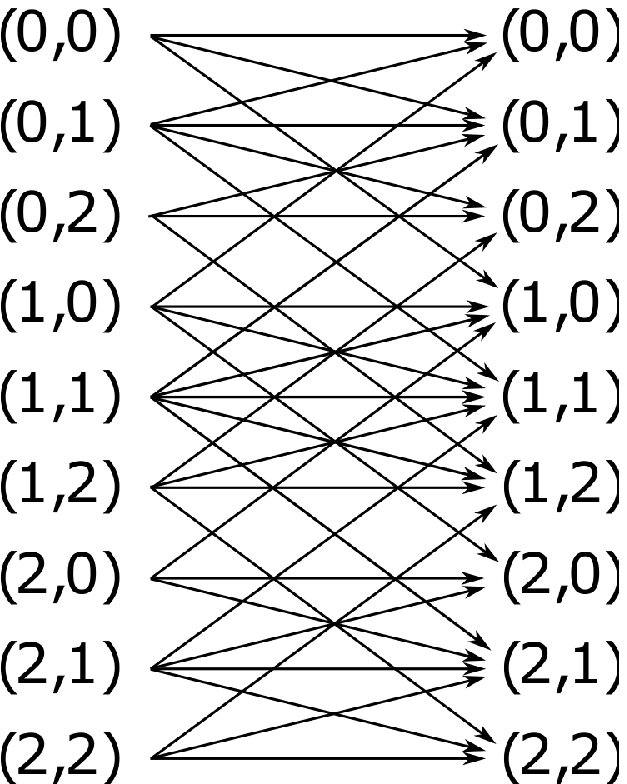}
        \subcaption{
            Channel transition graph of the NNE channel.
        }
        \label{fig:nne_channel}
    \end{minipage}
    \hspace{0.03\linewidth}
    \begin{minipage}[b]{0.48\linewidth}
        \centering
        \includegraphics[width=30mm]{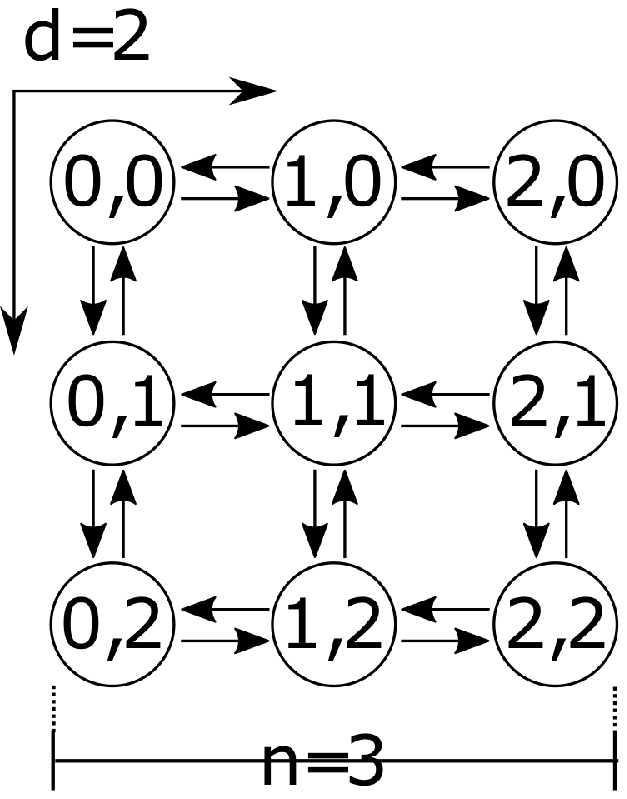}
        \subcaption{
		   Adjacency relation graph of the symbols.
        }
        \label{fig:nne_channe_adj}
    \end{minipage}
    \caption{
        An example of NNE channel.
		In this channel, a transmitted symbol is a  ($d=2$)-tuple of elements in $\{0, 1, 2\}$, i.e., $n = 3$.
        Black arrows in (b) represent possible symbol transitions.
    }
\end{figure}

%------------------------------------------------------------------------
%\subsection{Related works}
%------------------------------------------------------------------------
%
%In this subsection,
%we introduce some studies related to this paper.
%%------------------------------------------------------------------------
%\subsubsection{LMEを訂正・検知する符号の研究}
%%------------------------------------------------------------------------
%
%R. Ahlswede et al. studied codes over $q$-ary alphabets intended for the control of errors of $l$-levels~\cite{ahlswede2006}.
%They showed the number of maximum codewords of codes which have $n$ code length and is capable of correcting all asymmetric errors of level $l$.
% Additionally, they showed upper and lower bounds of the number of maximum codewords of codes which have $n$ code length and is capable of correcting all unidirectional errors of level $l$.
% There are the study of systematic codes capable of correcting all errors of level $1$~\cite{klove2011},
% and the study of the detection of limited-magnitude errors on $\mathbb{Z}_q$~\cite{elarief2013}.
%
%%------------------------------------------------------------------------
%\subsubsection{リー符号の研究}
%%------------------------------------------------------------------------
%

%------------------------------------------------------------------------
%\subsection{Outline}
%------------------------------------------------------------------------

In this paper, we will derive a lower bound of the zero error capacity of the NNE channels based on a result of the perfect Lee codes.
Furthermore, an upper bound of the zero error capacity will be derived from a feasible solution 
of an LP problem defined based on the confusion graph.

The outline of this paper is as follows.
In Section 2, we will introduce some definitions and basic results on the zero error capacity and Lee codes 
which are used in this paper.
In Section 3, the NNE channel will be defined. Afterwards, we will show 
upper and lower bounds on the zero error capacity of the NNE channels. 
Finally, a concise expression for the zero error capacity will be shown.
%=================
\section{Preliminaries}
%=================

In this section, we  introduce several definitions  and notation to be used in this paper.

%------------------------------------------------------------------------
\subsection{Confusion graph}
%------------------------------------------------------------------------
A Discrete Memoryless Channel (DMC) $W : \mathscr{X} \to \mathscr{Y}$ is defined by
the condition probability $W(y|x)$ that is the probability of receiving $y \in \mathscr{Y}$
when $x \in \mathscr{X}$ is sent over $W$.
The symbols $\mathscr{X}$ and $\mathscr{Y}$ represent the finite input and output alphabets, respectively.

The {\em confusion graph} of the channel $W$ is defined as follows.
\begin{Definition}
	Let $W : \mathscr{X} \to \mathscr{Y}$ be a DMC.
	The undirected graph $G_W = (V, E)$ is defined by
	\begin{align*}
		V &= \mathscr{X},\\
		E &= \Set{
			\{x_1, x_2 \} |
			\begin{array}{l}
				x_1, x_2 \in \mathscr{X}, \\
				x_1 \ne x_2, \\
				\sum_{y \in \mathscr{Y}} W(y | x_1) W(y | x_2) \ne 0
			\end{array}
		}.
	\end{align*}
	The graph $G_W$ is called the confusion graph of $W$.
\end{Definition}
\marginblock

Figure \ref{fig:nne_graph} shows the confusion graph of the channel shown in Fig.~\ref{fig:nne_channel}.

\begin{figure}[htb]
	\centering
	\includegraphics[width=40mm]{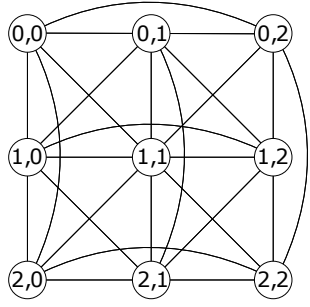}
	\caption{
			Confusion graph of the channel shown in Fig.~\ref{fig:nne_channel}
	}
	\label{fig:nne_graph}
\end{figure}

%------------------------------------------------------------------------
%\subsection{グラフの積}
%------------------------------------------------------------------------
The strong product is used in the definition of the zero error capacity.
The definition of the strong graph product is given as follows.
\begin{Definition}
	Let $G = (V_G, E_G)$ and $H = (V_H, E_H)$ be undirected graphs.
         Their strong product $G \boxtimes H = (V, E)$ is defined as
	\begin{align*}
		V &= V_G \times V_H,\\
		E &= \Set{
			\!\!
			\begin{array}{l}
				\{\ (x_1, y_1),\\
				\ \ (x_2, y_2)\ \}
			\end{array}
			| \!\!\!
			\begin{array}{l}
				(x_1, y_1), (x_2, y_2) \in V, \\
				(x_1, y_1) \ne (x_2, y_2), \\
				(x_1 = x_2) \lor (\{ x_1, x_2 \} \in E_G),\\
				(y_1 = y_2) \lor (\{ y_1, y_2 \} \in E_H)\\
			\end{array}\!\!\!\!\!
		}.
	\end{align*}
       The $n$-times power of $G$ is denoted by
	\begin{align*}
		G^n = \underbrace{G \boxtimes \cdots \boxtimes G}_{\text{$n$ times}}.
	\end{align*}
\end{Definition}
\marginblock

%------------------------------------------------------------------------
\subsection{Zero-error capacity}
%------------------------------------------------------------------------

We denote the {\em independence number} of an undirected graph $G$ by $\alpha(G)$.
The independence number is the size of the maximum independent set.
Shannon defined the zero error capacity of the discrete memoryless channel $W$
as follows.
\begin{Definition}
	Let $W : \mathscr{X} \to \mathscr{Y}$ be a DMC.
	The {\em zero error capacity} $\Gamma(W)$ is defined by
	\begin{align*}
		\Gamma(W) = \lim_{n \to \infty} \frac{\log_2 \alpha({G_W}^n)}{n}.
	\end{align*}
\end{Definition}
\marginblock

It is straightforward to show that the zero error capacity can be written as
\begin{align*}
	\Gamma(W) = \log_2 \Theta(G_W),
\end{align*}
where the capacity of the confusion graph $\Theta(G_W)$ is defined by
\begin{align*}
	\Theta(G) = \lim_{n \to \infty} \sqrt[n]{\alpha({G}^n)}.
\end{align*}

From the above definitions,  $\alpha(G)$ is clearly a lower bound on the capacity of $G$.
In this paper, we refer to this lower bound as the {\em Shannon's lower bound}.

The following theorem is known as an upper bound on the capacity of a graph．
\begin{Theorem}{\cite{shannon1956}\cite{lovasz1979}}\label{the:lp_upperbound}
	Let $\mathscr{C}(G)$ be the set of all the cliques in $G$ where $G = (V, E)$ is a given undirected graph.
	Consider the following LP problem:
	\begin{align}\label{eq:lp_shannon_main}
		\begin{array}{rl}
			\text{Minimize}   & \sum_{C \in \mathscr{C}(G)} q(C)\\
			\text{subject to} & \forall v \in V, \; \sum_{C \in \mathscr{C}(G), C \ni v} q(C) \geq 1 \\
			                  & \forall C \in \mathscr{C}(G), \; q(C) \geq 0,
		\end{array}
	\end{align}
	 where $q(C)$ is a real-valued variable corresponding to the clique $C \in \mathscr{C}(G)$.
	% on a mapping $q : \mathscr{C}(G) \to \mathbb{R}$:
	The optimal value $\alpha^*(G)$ of the above LP problem
	 is an upper bound on the capacity of $G$.
\end{Theorem}
\marginblock
In this paper, we call this bound the {\em LP upper bound}.

%------------------------------------------------------------------------
\subsection{Perfect Lee codes}
%------------------------------------------------------------------------

We introduce definitions of the Lee distance and the Lee codes
in this subsection according to ~\cite{jain2005,horak2009}.
In the following sections,
the residue class ring modulo $q$ is denoted by $\mathbb{Z}_q$,
and the coset including integer $a$ is denoted by $[a]_q$.

In order to define the Lee distance over $\mathbb{Z}_q$, we need to prepare
the two basic functions regarding the Lee distance.
\begin{Definition}
	The function
	$\psi_q : \mathbb{Z}_q \to \mathbb{Z}$ is defined as $\psi_q(a) = b$,
	 where $b$ is the minimum positive integer as $[b]_q = a$.
\end{Definition}
\marginblock

\begin{Definition}
	Let $x$ be an element in $\mathbb{Z}_q$.
	The absolute value $\abs{x}$ on $\mathbb{Z}_q$  is defined as
	\begin{align*}
		\abs{x} = \min(\psi_q(x),\; q - \psi_q(x)).
	\end{align*}
\end{Definition}
\marginblock

Based on the definition of the absolute value on $\mathbb{Z}_q$, the Lee distance is defined as follows.
\begin{Definition}
	The distance function $\rho : \mathbb{Z}_q^d \times \mathbb{Z}_q^d \to \mathbb{Z}^+$ is given by
	\begin{align*}
		\rho(\bvec{u},\bvec{v}) = \sum^d_{i=1} \abs{u_i - v_i},
	\end{align*}
	 where $\bvec{u} = (u_1, \ldots, u_d)$, $\bvec{v} = (v_1, \ldots, v_d)$.
	 The distance defined by $\rho$ is called the {\em Lee distance}.
\end{Definition}
\marginblock

We further introduce some definitions in order to define the perfect Lee codes.
\begin{Definition}
	Assume that $\bvec{x} \in \mathbb{Z}_q^d$ and $r \in \mathbb{Z}^+$ are given.
	 The Lee sphere $S_L(\bvec{x},r)$ is defined by
	\begin{align*}
		S_L(\bvec{x},r) = \left\{ \bvec{y} \in \mathbb{Z}_q^d \mid \rho(\bvec{x},\bvec{y}) \leq r \right\}.
	\end{align*}
\end{Definition}
\marginblock

\begin{Definition}\label{def:lee_code}
	Let $C$ be a subset of $\mathbb{Z}_q^d$.
	If
	\begin{align*}
		S_L(\bvec{x},e) \cap S_L(\bvec{y},e) = \emptyset
	\end{align*}
	 holds for any $\bvec{x},\bvec{y} \in C$ $(\bvec{x} \ne \bvec{y})$,
	 then $C$ is said to be an $e$-Lee error correcting code.
\end{Definition}
\marginblock

The perfect Lee code is one of the important tools for the proof of our main results.
\begin{Definition}
	Let $C$ be an $e$-Lee error correcting code of length $d$.
	If
	\begin{align*}
		\bigcup_{\bvec{x} \in C} S_L(\bvec{x},e) = \mathbb{Z}_q^d
	\end{align*}
	 holds, then $C$ is called a {\em perfect Lee code},
	 which is denoted by $\mathrm{PL}(d,e,q)$.
\end{Definition}
\marginblock

For given $d, q \in \mathbb{Z}^+$,
it is clear that $\mathrm{PL}(d,1,q)$ has $q^d / (2d+1)$ codewords.
 The following theorem proved by Albdaiwi et al.  \cite{albdaiwi2009}
 states the condition on $d$ and $q$ that guarantees the existence of $\mathrm{PL}(d,1,q)$.
\begin{Theorem}{\cite{albdaiwi2009}}\label{the:perfect_lee}
	For a given $d \in \mathbb{Z}^+$, assume that the prime factors of $2d+1$ are $p_1,p_2,\ldots, p_k$.
	A perfect Lee code $\mathrm{PL}(d,1,q)$ exists if and only if $(\prod_{i=1}^k p_i)|q$.
\end{Theorem}
\marginblock

From Theorem~\ref{the:perfect_lee}, we immediately have the following corollary.
\begin{Cor}\label{cor:lee_perfect}
A perfect Lee code $\mathrm{PL}(d,1,k(2d+1))$ exists for any $d, k \in \mathbb{Z}^+$.
\end{Cor}
\marginblock

As an example, Fig.~\ref{fig:lee_sphere} shows a perfect Lee code $C$ of length $d=2$ over $\mathbb{Z}_5$.
The code $C$ has 5-codewords such as
 $C = \left\{ \left(\overline{0},\overline{0}\right), \left(\overline{1},\overline{3}\right), \left(\overline{2},\overline{1}\right), \left(\overline{3},\overline{4}\right), \left(\overline{4},\overline{2}\right) \right\}$
where $\overline{0},\ldots, \overline{4}$ represent the cosets $[0]_5, \ldots, [4]_5$.
In this case, $C$ is a perfect Lee code $\mathrm{PL}(2,1,5)$.

\begin{figure}[htb]
	\centering
	\includegraphics[width=40mm]{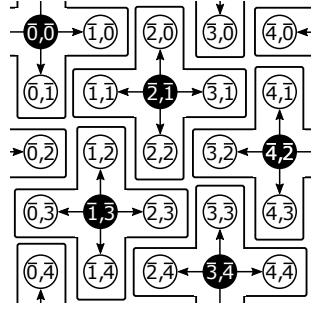}
	\caption{An example of $\mathrm{PL}(2,1,5)$. Black circles are the codewords of the perfect Lee code $C$.
         The crisscross shaped polygons represent Lee spheres with radius 1 around the codewords.}
	\label{fig:lee_sphere}
\end{figure}

%===================================
\section{Upper and lower bounds on zero error capacity of NNE Channel}
%===================================

In this section, we  firstly define the NNE channel model.
Afterward, upper and lower bounds of the capacity of these channels are derived.

%------------------------------------------------------------------------
\subsection{NNE channel}
%------------------------------------------------------------------------
The $L_1$ distance is defined as usual way.
\begin{Definition}
	The $L_1$ distance function $D_{L_1} : \mathbb{Z}^d \times \mathbb{Z}^d \to \mathbb{Z}^+$ is defined as
	\begin{align*}
		D_{L_1}(\bvec{u}, \bvec{v}) = \sum_{i=1}^{d} \abs{u_i - v_i},
	\end{align*}
	 where $d \in\mathbb{Z}^+$ and $\bvec{u} = (u_1, \ldots, u_d)$, $\bvec{v} = (v_1, \ldots, v_d)$.
\end{Definition}
\marginblock

The NNE channel model is defined as follows.
\begin{Definition}\label{nne_channel}
	Let $W : [0,n-1]^d \to [0,n-1]^d$ be a DMC.
	If
	\begin{align*}
		D_{L_1}(\bvec{x}, \bvec{y}) \le 1 \iff W(\bvec{y}|\bvec{x}) \ne 0
	\end{align*}
	 holds for any $\bvec{x}, \bvec{y} \in [0,n-1]^d$,
	 then $W$ is said to be the $d$-dimensional {\em NNE channel} with the $n$-level alphabet abbreviated 
	 as the $(d,n)$-NNE channel.
\end{Definition}
\marginblock

From this definition, 
the channel shown in Fig.~\ref{fig:nne_channel} should be denoted by the $(d=2,n=3)$-NNE channel.

The confusion graph of the $(d,n)$-NNE channel, called the {\em NNE graph}, 
is unique because of the definition of the confusion graphs and
the definition of the $(d,n)$-NNE channel.

%------------------------------------------------------------------------
\subsection{Lower bound on zero error capacity}\label{sec:lower_bound}
%------------------------------------------------------------------------

In this subsection, we show a lower bound on the capacity of the NNE graphs.
The following lemmas are required for the proof of the lower bound.

\begin{Lemma}\label{lem:graph_equiv_L1}
	Let $G_W = (V, E)$ be the NNE graph of $W$.
	The inequality $D_{L_1}(\bvec{x}_1, \bvec{x}_2) \leq 2$ holds if and only if $\{\bvec{x}_1, \bvec{x}_2\} \in E$
	 where $\bvec{x}_1, \bvec{x}_2 \in V$.
\end{Lemma}

\begin{proof}
	From the definition of the NNE channels,
	%an input symbol $\bvec{x} \in [0,n-1]^d$ outputs a symbol in
 	when $\bvec{x} \in [0,n-1]^d$ is transmitted, the receiver side receives a symbol contained in
	\begin{align*}
		S_{L_1}(\bvec{x}) = \left\{ \bvec{y} \in [0,n-1]^d \mid D_{L_1}(\bvec{x}, \bvec{y}) \leq 1 \right\}.
	\end{align*}
	From the definition of confusion graphs,
	\begin{align*}
		\hspace{-2mm}\sum_{\bvec{y} \in [0, n - 1]^d} W(\bvec{y}|\bvec{x}_1)W(\bvec{y}|\bvec{x}_2) \ne 0
		\iff \{\bvec{x}_1, \bvec{x}_2\} \in E
	\end{align*}
	holds for any $\bvec{x}_1, \bvec{x}_2 \in [0, n-1]^d$ ($\bvec{x}_1 \ne \bvec{x}_2$).
	Therefore, we have
	\begin{align}\label{eq:l1_adjacent_1}
		S_{L_1}(\bvec{x}_1) \cap S_{L_1}(\bvec{x}_2) \ne \emptyset \iff \{\bvec{x}_1, \bvec{x}_2\} \in E.
	\end{align}
	The intersection $S_{L_1}(\bvec{x}_1) \cap S_{L_1}(\bvec{x}_2)$ can be written by
	\begin{align*}
		\begin{array}{l}
			S_{L_1}(\bvec{x}_1) \cap S_{L_1}(\bvec{x}_2)
			 = \Set{\!
				\bvec{y} \in [0,n-1]^d \;| \!\!
				\begin{array}{l}
					D_{L_1}(\bvec{x}_1, \bvec{y}) \leq 1,\\
					D_{L_1}(\bvec{x}_2, \bvec{y}) \leq 1
				\end{array}\!\!\!
			}.
		\end{array}
	\end{align*}
	Therefore, we obtain
	\begin{align}\label{eq:l1_adjacent_2}
		\begin{array}{l}
			S_{L_1}(\bvec{x}_1) \cap S_{L_1}(\bvec{x}_2) \ne \emptyset
			\iff D_{L_1}(\bvec{x}_1, \bvec{x}_2) \leq 2.
		\end{array}
	\end{align}
	From (\ref{eq:l1_adjacent_1}),(\ref{eq:l1_adjacent_2}), we obtain
	\begin{align*}
		D_{L_1}(\bvec{x}_1, \bvec{x}_2) \leq 2 \iff \{\bvec{x}_1,\bvec{x}_2\} \in E.
	\end{align*}
\end{proof}
\marginblock

The following lemma shows the relation between a distance over $\mathbb{Z}$ and a distance over $\mathbb{Z}_q$.

\begin{Lemma}\label{lem:psi_inequality}
	For any $u$ and $v$ in $\mathbb{Z}_q$,
	\begin{align*}
		\abs{u-v} \leq \abs{\psi_q(u) - \psi_q(v)}
	\end{align*}
	holds.
\end{Lemma}

\begin{proof}
	For any $u$ and $v$ in $\mathbb{Z}_q$,  we have
	\begin{align*}
		\psi_q(u-v) = \begin{cases}
			\psi_q(u) - \psi_q(v),       & \text{if $u \geq v$},\\
			q - (\psi_q(u) - \psi_q(v)), & \text{otherwise}.
		\end{cases}
	\end{align*}
	 Therefore, if $u \geq v$, then
	\begin{align*}
		\abs{u-v} = \min\bigl(\psi_q(u-v), \; q - \psi_q(u-v)\bigr)
		          = \min\bigl(\psi_q(u) - \psi_q(v), \; q - (\psi_q(u) - \psi_q(v))\bigr)
	\end{align*}
	holds.
%	 That is, if $u \geq v$, then $\abs{u-v} \leq \abs{\psi_q(u) - \psi_q(v)}$ because $\abs{\psi_q(u) - \psi_q(v)} = \psi_q(u) - \psi_q(v)$.
	 On the other hand,  if $u < v$, then we get
	\begin{align*}
		\abs{u-v} = \min\bigl(\psi_q(u-v), \; q - \psi_q(u-v)\bigr)
		          = \min\bigl(q - (\psi_q(u) - \psi_q(v)), \; \psi_q(v) - \psi_q(u)\bigr).
	\end{align*}
%	 That is, if $u < v$, then $\abs{u-v} \leq \abs{\psi_q(u) - \psi_q(v)}$ from $\abs{\psi_q(u) - \psi_q(v)} = \psi_q(v) - \psi_q(u)$.
\end{proof}

The next lemma will play a key role in the proof of the lower bound.

\begin{Lemma}\label{lem:lower_bound}
	There exists a subset $V$ in $[0,n-1]^d$ satisfying the following two conditions:
	\begin{align*}
		\abs{V} \geq \ceil{\frac{n^d}{2d+1}}
	\end{align*}
	and
	\begin{align*}
		D_{L_1}(\bvec{u},\bvec{v}) \geq 3
	\end{align*}
	for any $\bvec{u},\bvec{v} \in V$ ($\bvec{u}\ne\bvec{v}$).
\end{Lemma}

\begin{proof}
	Let  $C$ be a perfect Lee code $\mathrm{PL}(d,1,q)$
	 where $q$~($n \leq q$) is a multiple of $2d+1$.
	The existence of such a perfect Lee code $C$ is guaranteed by Corollary~\ref{cor:lee_perfect}.
	Let $U : \mathbb{Z}_q^d \to 2^{\mathbb{Z}_q^d}$ be the set defined as
	\begin{align*}
		U(\bvec{x}) = \Set{\!
			\bvec{u} \in C | \!\!
			\begin{array}{l}
				i\in [1,d],\\
				u_i \in \Set{x_i + [j]_q | j \in [0,n-1]}
			\end{array} \!\!\!\!
		},
	\end{align*}
	 where $\bvec{u} = (u_1, \ldots, u_d)$, $\bvec{x} = (x_1, \ldots, x_d)$.
	From the definition of $U(\bvec{x})$, it is clear that 
	the identity 
	\begin{align*}
		\sum_{\bvec{x} \in \mathbb{Z}_q^d} \mathbb{I}[\bvec{c} \in U(\bvec{x}) ] = n^d
	\end{align*}
	holds for all $\bvec{c} \in C$. The function $\mathbb{I}[condition]$ takes the value one if
	$condition$ is true, otherwise it takes the value zero.
	From the above identity, we immediately have
	\begin{align}
		\sum_{\bvec{x} \in \mathbb{Z}_q^d} \abs{U(\bvec{x})}
		  =  \sum_{\bvec{x} \in \mathbb{Z}_q^d} \sum_{\bvec{c} \in C} \mathbb{I}[\bvec{c} \in U(\bvec{x})]
		  = n^d\abs{C}. \label{eq:v_exist_1}
	\end{align}
%	Therefore, we have the identity:
%	\begin{align}\label{eq:v_exist_1}
%		n^d\abs{C} = \sum_{\bvec{x} \in \mathbb{Z}_q^d} \abs{U(\bvec{x})}.
%	\end{align}
%	The set $U(\bvec{x})$ is called the area starting $\bvec{x}$.
%	For any $\bvec{w} \in C$, the number of \textcolor{red}{the areas containing $\bvec{w}$
%	% possibilities to have $\bvec{x}$ satisfying $U(\bvec{x}) \ni \bvec{w}$
%	 is $n^d$.
%	The sum of the size of all the areas equals 
%	the sum of the number of the areas containing $\bvec{w}$ for all $\bvec{w} \in C$.}
	The identity (\ref{eq:v_exist_1}) can be rewritten as
	\begin{align}\label{eq:v_exist_2}
		n^d \cdot \frac{q^d}{2d+1} = \sum_{\bvec{x} \in \mathbb{Z}_q^d} \abs{U(\bvec{x})},
	\end{align}
	 because $C$ is $\mathrm{PL}(d,1,q)$.
	We now assume that
	\begin{align} \label{eq:assumption}
		\abs{U(\bvec{x})} \leq \ceil{\frac{n^d}{2d+1}} - 1
	\end{align}
	 holds for any $\bvec{x} \in \mathbb{Z}_q^d$.
	By applying this assumption to (\ref{eq:v_exist_2}), we get the inequalities
	\begin{align*}
		n^d \cdot \frac{q^d}{2d+1}
		 \leq q^d \left(\ceil{\frac{n^d}{2d+1}} - 1\right)
		 < q^d \left(\frac{n^d}{2d+1}\right).
	\end{align*}
	This inequality contradicts the assumption (\ref{eq:assumption}). Therefore, it can be concluded that
	 there exists $\bvec{x} \in \mathbb{Z}_q^d$ which satisfies
	\begin{align}\label{eq:v_exist_3}
		\abs{U(\bvec{x})} > \ceil{\frac{n^d}{2d+1}} - 1.
	\end{align}
	% For any $\bvec{x} = (x_1, \ldots, x_d)$ \textcolor{red}{satisfying the inequality (\ref{eq:v_exist_3})},
	The function $f : U(\bvec{x}) \to [0,n-1]^d$ is defined as
	\begin{align*}
		f(\bvec{u}) = ( \psi_{q}(u_1 - x_1), \ldots, \psi_{q}(u_d - x_d)),
	\end{align*}
	 where $\bvec{u} = (u_1, \ldots, u_d)$.
	Finally, $V \subset [0,n-1]^d$ is defined by
	\begin{align*}
		V = \{ f(\bvec{u}) \mid \bvec{u} \in U(\bvec{x}) \}.
	\end{align*}
	Since $f$ is a one-to-one function, we have $\abs{U(\bvec{x})} = \abs{V}$.
	For any $\bvec{u} = (u_1, \ldots, u_d)$, $\bvec{v} = (v_1, \ldots, v_d)$ $\in U(\bvec{x})$
	 and $i \in [1,d]$, the following inequality
	\begin{align*}
		\quad \abs{u_i - v_i} \leq \abs{\psi_{q}(u_i - x_i) - \psi_{q}(v_i - x_i)}
	\end{align*}
	 holds due to Lemma~\ref{lem:psi_inequality}. This inequality directly leads to
	\begin{align*}
		\sum_{i=1}^{d} \abs{u_i - v_i} \leq \sum_{i=1}^{d} \abs{\psi_{q}(u_i - x_i) - \psi_{q}(v_i - x_i)}.
	\end{align*}
	The definitions of distance function $\rho,D_{L_1}$ can transform this inequality to
	\begin{align}\label{eq:v_exist_4}
		\rho(\bvec{u}, \bvec{v}) \leq D_{L_1}(f(\bvec{u}), f(\bvec{v})).
	\end{align}
	The inequalities (\ref{eq:v_exist_3}) and (\ref{eq:v_exist_4}) imply the claim of this lemma.
\end{proof}
\marginblock

As an example of Lemma~\ref{lem:lower_bound}, Fig.~\ref{fig:lower_bound} shows an instance of $V \subset [0, 4]^2$ 
satisfying the conditions.
The cosets $[0]_5, \ldots, [4]_5$ are denoted by $\overline{0}, \ldots, \overline{4}$, respectively.
The subset $U((\overline{4},\overline{1}))$ equals to $\{ (\overline{4},\overline{1}), (\overline{1},\overline{3}) \}$ shown in Fig.~\ref{fig:lower_bound1}.
Therefore the subset $V$ corresponding to $U((\overline{4},\overline{1}))$ equals to $\{ (0, 1), (2, 2) \}$ shown in Fig.~\ref{fig:lower_bound2}.
This subset satisfies the conditions in Lemma~\ref{lem:lower_bound}.

\begin{figure}[htb]
    \begin{minipage}[b]{0.48\linewidth}
        \centering
        \includegraphics[width=40mm]{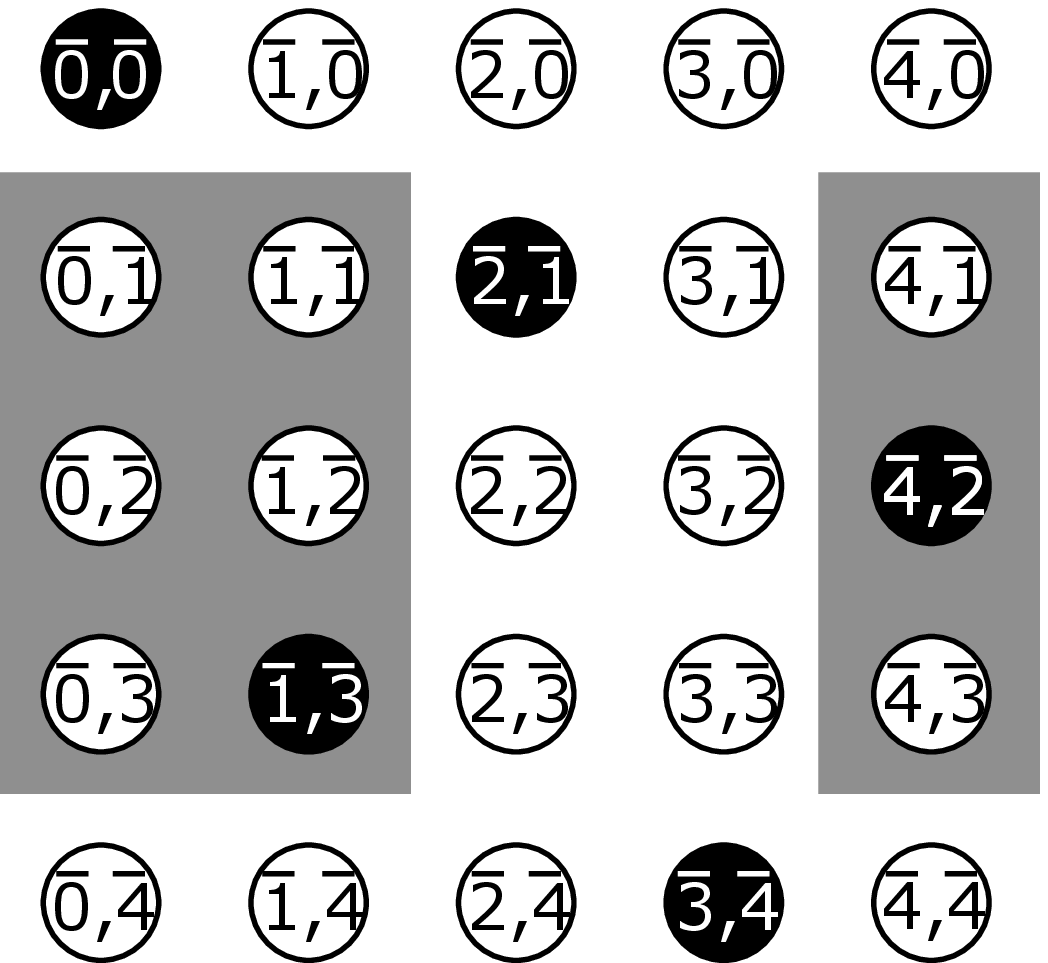}
        \subcaption{
				Black circles in the gray area represent $U((\overline{4},\overline{1}))$.
				\\\quad
        }
        \label{fig:lower_bound1}
    \end{minipage}
    \hspace{0.03\linewidth}
    \begin{minipage}[b]{0.48\linewidth}
        \centering
        \includegraphics[width=40mm]{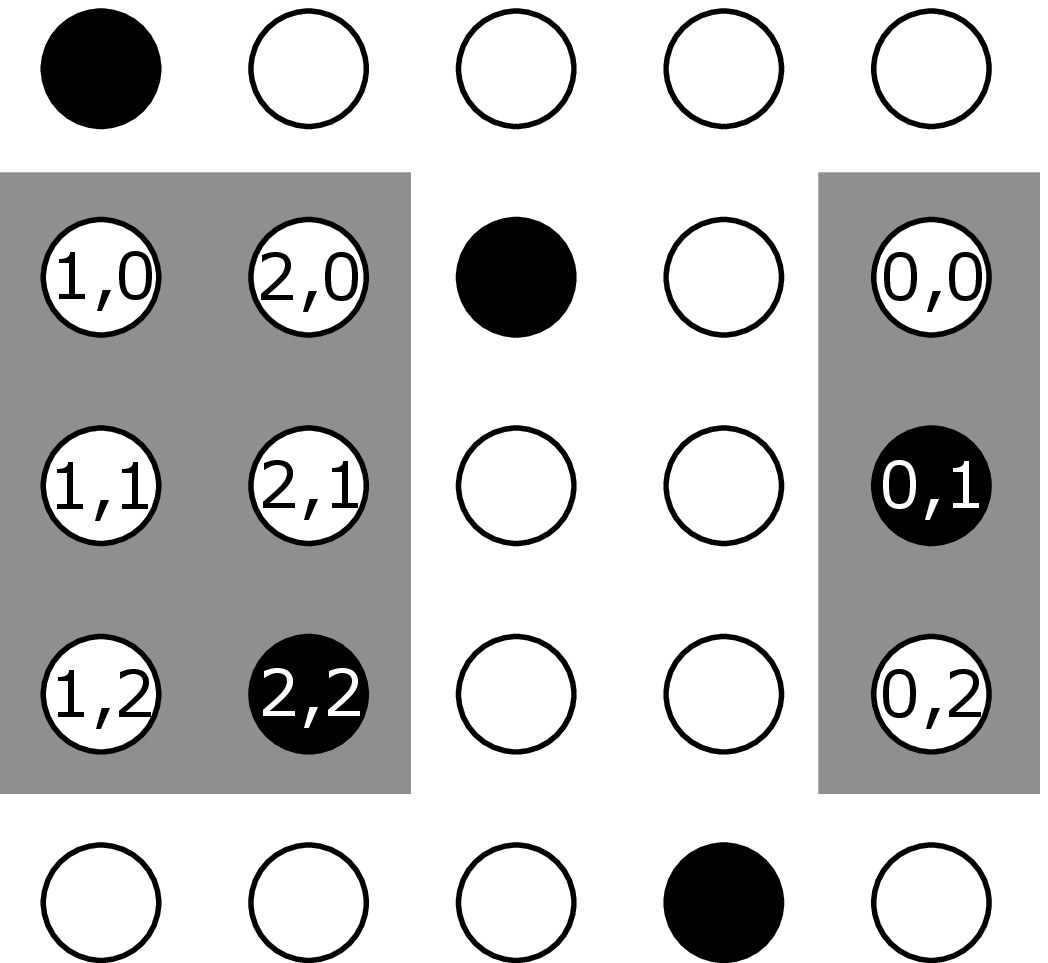}
        \subcaption{
				Black circles in the gray area represent $V$ corresponding to $U((\overline{4},\overline{1}))$.
        }
        \label{fig:lower_bound2}
    \end{minipage}
    \caption{
			An example of $V \subset [0, 4]^2$ satisfing the conditions in Lemma~\ref{lem:lower_bound}.
			Black circles are the codewords of the perfect Lee code $\mathrm{PL}(2,1,5)$.
    }
	\label{fig:lower_bound}
\end{figure}

From the above lemmas, we are ready to prove the following theorem stating
a lower bound on the capacity of the NNE graphs. This theorem is our main contribution in this paper.

\begin{Theorem}\label{the:lower_bound}
	Let $G_W$ be the confusion graph of the $(d,n)$-NNE channel $W$.
	The inequality
	\begin{align}\label{eq:lower_bound}
		\ceil{\frac{n^d}{2d+1}} \leq \Theta(G_W)
	\end{align}
	holds.
\end{Theorem}

\begin{proof}
	Let $V$ be a subset of $[0,n-1]^d$ given by Lemma~\ref{lem:lower_bound}.
	From Lemma~\ref{lem:graph_equiv_L1}, any two vertices in $V$ are not adjacent to each other in $G_W$.
	That is, $V$ is an independent set of $G_W$.
	The size of $V$ is $\ceil{n^d / (2d+1)}$ and can be considered as a lower bound on $\alpha(G_W)$.
	By using the Shannon's lower bound, we have the inequality (\ref{eq:lower_bound}).
\end{proof}
\marginblock

%------------------------------------------------------------------------
\subsection{Upper bound on zero error capacity}\label{sec:upper_bound}
%------------------------------------------------------------------------

In this subsection, we show an upper bound on the capacity of the NNE graphs
 by considering a certain feasible solution of the LP problem in Theorem~\ref{the:lp_upperbound}.
In order to obtain this solution, we introduce several definitions.

\begin{Definition}
	Let $G_W$ be the confusion graph of the $(d,n)$-NNE channel $W$.
	Let $\nu : \mathbb{Z}^d \to \mathbb{Z}$ be the function defined by
	\begin{align*}%\label{eq:nu}
		\nu(\bvec{u}) = \sum_{i=1}^d s(u_i), \quad
		s(x) = \begin{cases}
			1, & \text{if $x \in \{0, n - 1\}$},\\
			0, & \text{otherwise},
		\end{cases}
	\end{align*}
	 where $\bvec{u} = (u_1, \ldots, u_d)$.
	A vertex $\bvec{x}$ of $G_W$ is called an inner vertex if $\nu(\bvec{x}) = 0$ holds;
	 otherwise $\bvec{x}$ is called an outer vertex.
\end{Definition}
\marginblock

\begin{Definition}
	Let $G_W$ be the confusion graph of the $(d,n)$-NNE channel $W$.
	Let $\gamma : [0,n-1]^d \to \mathscr{C}(G_W)$ be the function defined as
	\begin{align*}%\label{eq:c_x}
		\gamma(\bvec{x}) = \left\{ \bvec{y} \in [0,n-1]^d \mid D_{L_1}(\bvec{x},\bvec{y}) \leq 1 \right\}.
	\end{align*}
	For any vertex $\bvec{x}$ in $[0,n-1]^d$,
	 an adjacency clique $C$ in $\gamma(\bvec{x})$ is called an inner adjacency clique if $\bvec{x}$ is an inner vertex;
	 otherwise $C$ is called an outer adjacency clique.
\end{Definition}
\marginblock

We show an example of inner and outer vertices and an adjacency clique.
On the confusion graph of the $(d=2, n=3)$-NNE channel shown in Fig.~\ref{fig:nne_graph},
the vertex $(1, 1)$ is an inner vertex because $\nu((1,1)) = 0$.
On the other hand, the vertices $(1, 0), (0, 0)$ are outer vertices because $\nu((1, 0)), \nu((0,0))$ equal to $1, 2$, respectively.
Hence, the adjacency clique in $\gamma((0, 0))$ is an outer adjacency clique and consists of the vertices $(0,0), (0,1), (1,0)$.

The following lemma states the existence of a certain feasible solution of the LP problem (\ref{eq:lp_shannon_main}).

\begin{Lemma}\label{lem:feasible}
	Let $G_W$ be the confusion graph of the $(d,n)$-NNE channel $W$.
	Let $q : \mathscr{C}(G_W) \to \mathbb{R}$ be the function defined as
	\begin{align*}%\label{eq:feasible}
		q(C) = \begin{cases}
			\frac{1}{2d+1},     & \text{if $C$ is an inner adjacency clique},\\
			1,                  & \text{if $C$ is an outer adjacency clique},\\
			0, & \text{otherwise}.
		\end{cases}
	\end{align*}
	The function $q$ is a feasible solution of the LP problem (\ref{eq:lp_shannon_main}).
\end{Lemma}

\begin{proof}
	For any outer vertex $\bvec{x}$, it is trivial that
	\begin{align}\label{eq:constraints}
		\sum_{\substack{C \in \mathscr{C}(G),\\C \ni \bvec{x}}} q(C) \geq 1
	\end{align}
	 holds.
	Since the size of an inner adjacency clique is $2d+1$,
	 the inequality (\ref{eq:constraints}) holds for any inner vertex as well.
	Therefore, the function $q$ satisfies constraints of the LP problem (\ref{eq:lp_shannon_main}).
\end{proof}
\marginblock

From the above lemmas, 
we can prove the following theorem for an upper bound on the capacity of the NNE graphs.

\begin{Theorem}\label{the:upper_bound}
	Let $G_W$ be the confusion graph of the $(d,n)$-NNE channel $W$.
	The inequality 
	\begin{align*}
		\Theta(G_W) \leq \alpha'(G_W)
	\end{align*}
	holds where
	\begin{align*}
		\alpha'(G_W) = (n-2)^d \cdot \frac{1}{2d+1} + \bigl( n^d - (n-2)^d \bigr).
	\end{align*}
\end{Theorem}

\begin{proof}
	Since the number of all the inner vertices in $G_W$ is $(n-2)^d$,
	 the feasible solution in Lemma~\ref{lem:feasible} provides the objective 
	 function value equal to $\alpha'(G_W)$. The LP upper bound leads to the claim of
	 the theorem.
\end{proof}
\marginblock

%------------------------------------------------------------------------
\subsection{Expression of zero-error capacity with asymptotically vanishing term}
%------------------------------------------------------------------------

From Theorem~\ref{the:lower_bound} and Theorem~\ref{the:upper_bound},
we can prove the following theorem stating a concise expression for the zero error capacity of the NNE channels.

\begin{Theorem}\label{the:zero_error}
	Let $W$ be the $(d,n)$-NNE channel.
	For any $d \in \mathbb{Z}^+$, the zero error capacity $\Gamma(W)$ can be written as
	\begin{align*}
		\Gamma(W) = d \log_2 n + \log_2 \left(\frac{1}{2d+1} + \mathcal{O}\left(n^{-1}\right)\right).
	\end{align*}
\end{Theorem}

\begin{proof}
	Let $G_W$ be a confusion graph of $W$.
	From Theorem~\ref{the:lower_bound}, the inequality (\ref{eq:lower_bound}) holds.
	This inequality can be rewritten as
	\begin{align}\label{eq:O_lower_bound}
		\frac{n^d}{2d+1} \leq \Theta(G_W).
	\end{align}
	From Theorem~\ref{the:upper_bound}, the inequality $\Theta(G_W) \leq \alpha'(G_W)$ holds.
	Therefore, $\alpha'(G_W)$ can be transformed as follows:
	\begin{align*}
		\alpha'(G_W) &= \frac{(n-2)^d}{2d+1} + \bigl( n^d - (n-2)^d \bigr)\\
		           &= \frac{n^d}{2d+1} + \frac{2d \left( n^d - (n-2)^d \right)}{2d+1} \\
		           &= \frac{n^d}{2d+1} - \frac{2d \left( (-2) d n^{d-1} + \cdots + (-2)^d \right)}{2d+1}.
	\end{align*}
	%\begin{align*}
	%	\alpha'(G_W) &= (n-2)^d \cdot \frac{1}{2d+1} + \bigl( n^d - (n-2)^d \bigr)\\
	%	           &= \frac{n^d}{2d+1} + \left( n^d - (n-2)^d \right) \frac{2d}{2d+1} \\
	%	           &= \frac{n^d}{2d+1} - \left( (-2) d n^{d-1} + \cdots + (-2)^d \right)\frac{2d}{2d+1}.
	%	\end{align*}
	Since the inequality
	\begin{align*}
		\frac{2d \left( (-2) d n^{d-1} + \cdots + (-2)^d \right)}{2d+1} < 0
	\end{align*}
	 holds for any $n,d \in \mathbb{Z}^+$, $\alpha'(G_W)$ can be written as
	\begin{align}\label{eq:O_upper_bound}
		\alpha'(G_W) = \frac{n^d}{2d+1} + \mathcal{O}\left(n^{d-1}\right).
	\end{align}
	From (\ref{eq:O_lower_bound}) and (\ref{eq:O_upper_bound}), we thus have
	\begin{align*}
		\Theta(G_W) = n^d \left( \frac{1}{2d+1} + \mathcal{O}\left(n^{-1}\right)\right),
	\end{align*}
	and an expression for the zero error capacity:
	\begin{align*}
		\Gamma(W) &= \log_2 \left( n^d \left( \frac{1}{2d+1} + \mathcal{O}\left(n^{-1}\right)\right) \right)\\
		          &= d \log_2 n + \log_2 \left(\frac{1}{2d+1} + \mathcal{O}\left(n^{-1}\right)\right).
	\end{align*}
\end{proof}
\marginblock
The error term in the above expression vanishes at the asymptotic limit $n \rightarrow \infty$.

%==============
\section{Conclusion}
%==============

In this paper, we studied the zero error capacity of the NNE channels.
We derived a lower bound of the capacity of the NNE graphs
based on the idea that the zero error communication problem over 
the NNE channels can be regarded as a packing problem of $d$-dimensional polyominos.
The perfect Lee codes play a crucial role for proving the existence of a dense packing.
Moreover, we showed an upper bound of the capacity of the NNE graphs
from a certain feasible solution of the LP problem given by Shannon.
Finally, from these upper and lower bounds, we obtained 
a concise expression of the zero error capacity of the NNE channels
with an asymptotically (i.e., $n \rightarrow \infty$) vanishing error term $\mathcal{O}(n^{-1})$.

\section*{Acknowledgment}

%This work was supported by JSPS Grant-in-Aid for Scientific Research Grant Number 16K14267.

\end{document}